\documentclass[onefignum,onetabnum]{siamart171218}

\usepackage{epsfig}
\usepackage{latexsym}
\usepackage{amsmath}
\usepackage{graphicx}
\usepackage{amssymb}

\newtheorem{thm}{Theorem}

\newtheorem{prop}{Proposition}
\newtheorem{cor}{Corollary}

\newcommand{\ignore}[1]{}

\newcommand{\R}{\mathcal{R}}
\newcommand{\Q}{\mathcal{Q}}
\newcommand{\A}{\mathcal{A}}
\newcommand{\D}{\mathcal{D}}
\newcommand{\C}{\mathcal{C}}

\newcommand{\cL}{\mathcal{L}}
\newcommand{\RR}{\mathbb{R}}

\newcommand{\citep}{\cite}

\title{Tractable models of self-sustaining autocatalytic networks}

\author{Mike Steel\thanks{Biomathematics Research Centre, University of Canterbury, Christchurch, New Zealand (\email{mike.steel@canterbury.ac.nz}).}
\and Wim Hordijk\thanks{Institute for Advanced Study, University of Amsterdam, The Netherlands (\email{wim@WorldWideWanderings.net}).}}

\begin{document}

\maketitle

\begin{abstract}
Self-sustaining autocatalytic networks play a central role in living systems, from metabolism at the origin of life, simple RNA networks, and the modern cell, to ecology and cognition. A collectively autocatalytic network that can be sustained from an ambient food set is also referred to more formally as a `Reflexively Autocatalytic F-generated' (RAF) set. In this paper, we first investigate a simplified setting for studying RAFs, which are nevertheless relevant to real biochemistry and allows for a more exact mathematical analysis based on graph-theoretic concepts. This, in turn, allows for the development of efficient (polynomial-time) algorithms for questions that are computationally NP-hard in the general RAF setting. We then show how this simplified setting for RAF systems leads naturally to a more general notion of RAFs that are `generative' (they can be built up from simpler RAFs) and for which efficient algorithms carry over to this more general setting. Finally, we show how classical RAF theory can be extended to deal with ensembles of catalysts as well as the assignment of rates to reactions according to which catalysts (or combinations of catalysts) are available.
\end{abstract}



\section{Introduction}

A central property of the chemistry of living systems is that they combine two basic features: (i) the ability to survive on an ambient food source, and (ii) each biochemical reaction in the system requires only reactants and a catalyst that are provided by other reactions in the system (or are present in the food set).   The notion of a self-sustaining `collectively autocatalytic set' tries to capture these basic features formally, and was pioneered by Stuart Kauffman \citep{kau1, kau2},  who investigated a simple binary polymer model to address questions that relate to the origin of life.  The notion of a collectively autocatalytic set was subsequently formalised more precisely as `Reflexively Auto-catalytic and F-generated' (RAF) sets (defined shortly) and explored by others \citep{vas, hor17}.
RAFs are related to other notions such as Rosen's (M;R) systems \citep{jar}, and `organisations' in Chemical Organisation Theory \citep{hor17a}.   The application of RAFs has expanded beyond toy polymer models to analyse both real living systems (e.g. the metabolic network of {\em Escherichia  coli} \citep{sou}) and simple autocatalytic RNA systems that have  recently been generated in laboratory studies  by \cite{vai}.

The generality of RAF theory also means that a `reaction' need not refer specifically to a chemical reaction, but  to any process in which `items' are combined and transformed into new `items', and where similar `items' that are not used up in the process  facilitate (or `catalyse') the process. This has  led to  application of RAF theory to processes beyond biochemistry, including biodiversity \citep{gat}, cognitive psychology \citep{gab},  and (more speculatively) economics \citep{hor17}. 

In this paper, we show how RAF theory can be developed further to:
\begin{itemize}
\item
provide an exact and tractable characterisation of RAFs and subRAFs when reactants involve just food molecules;
\item
extend this last concept to general catalytic reaction networks by defining a new type of RAF (`generative') which couples realism with tractability; and
\item
include reaction rates into RAF theory and show that an optimal RAF can be calculated in polynomial time.
\end{itemize}
  
We begin with some definitions.

\subsection{Catalytic reaction systems (CRSs)}

A {\em catalytic reaction system} (CRS) consists of a set $X$ of `molecule types', a set $\R$ of `reactions', an assignment  $C$ describing  which molecule types catalyse which reactions, and a subset $F$ of $X$
consisting of a `food set' of basic building block molecule types  freely available from the environment.  Here, a `reaction' refers to a process that takes one or more molecule types (the `reactants') as input and produces one or more molecule types 
as output (`products'). $C$ can be viewed as a subset of $X \times \R$.   

A CRS can be represented mathematically  in two essentially equivalent ways:  Firstly as a directed graph with two types of vertices (corresponding to molecule types (some of which lie in the food set $F$) and reactions) and two types of arcs (arcs from molecule types into and out of reaction vertices (as reactants and products, respectively, and arcs from molecule types to reactions to denote catalysis). Fig.~\ref{fig1} provides a simple example of a CRS represented in this way.

Alternatively, one can list the reactions explicitly, writing each in the form $$r: A \xrightarrow{c_1, c_2, \ldots} B,$$ where $A$ denotes the set of reactants of reaction $r$, $B$ the set of products of $r$, and $c_1, c_2, \ldots$ are the possible catalysts for $r$.  For example, for $r_2$ and $r_3$ in the CRS of Fig.~\ref{fig1} we write:
 $$r_2: f_2+f_3  \xrightarrow{p_3} p_2, $$ 
 $$r_3: f_3+f_4 \xrightarrow{p_2} p_3,$$
to denote that  $r_2$ is catalysed by $p_3$ and $r_3$ is catalysed by $p_2$.

 \begin{figure}[htb]
\centering
\includegraphics[scale=1.1]{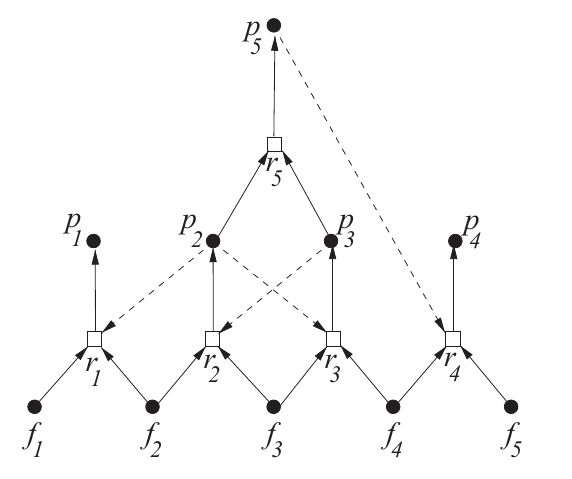}
\caption{A simple CRS  involving five reactions (box vertices), $r_1, \ldots, r_5$, and ten molecule types (round vertices) namely, a food set $F = \{f_1, \ldots, f_5\}$ together with  five other molecule types $p_1, \ldots, p_5$).  Catalysation arcs are shown as dashed arrows.  
In this CRS, the subsets $\{r_2, r_3\}$ and $\{r_1, r_2, r_3\}$ are the two RAFS (the former is an irrRAF, the latter is the maxRAF). In this example, each reaction has exactly two reactants, one product, and at most one catalyst,  however a CRS can have reactions with an arbitrary number of reactants, products and possible catalysts.
 }
\label{fig1}
\end{figure}

\bigskip

{\bf Definitions (RAFs, maxRAF)}.

Given a CRS $\Q=(X,\R, C, F)$, a subset $\R'$ of $\R$ is a said to be a {\em RAF} for $\Q$ if $\R'$ is nonempty and satisfies the following two conditions.
\begin{itemize}
\item {\em Reflexively autocatalytic} (RA): each reaction $r \in \R'$  is catalysed by at least one molecule type that is either present in the food set or is generated by another reaction in $\R'$.
\item {\em Food-generated (F)}: for each reaction $r\in \R'$, each reactant of $r$ is either present in the food set $F$ or can be generated by a sequence of reactions from $\R'$, each of which has each of its reactants present either in the food set or as a product of an earlier reaction in the sequence.
\end{itemize}
In other words, a RAF is a subset of reactions that is both self-sustaining (from the food set) and collectively autocatalytic.  In forming a RAF from the food set, some reactions may initially need to proceed uncatalysed (and thereby at a lower rate) 
but once formed every reaction in the RAF will be catalysed.  A simple example of a RAF is the pair of reactions $\{r_3, r_4\}$ shown in the CRS of Fig.~\ref{fig1}. 
 Note that in this example either $r_3$ or $r_4$ must first proceed uncatalysed, but once one reaction has occurred, the system continues with both reactions catalysed.

The food-generated condition (F) can also be formalised as follows:  For an arbitrary subset $\R'$ of $\R$ let 
${\rm cl}_{\R'}(F)$ be the (unique) minimal subset $W$ of $X$ that contains $F$ and has the property that if $r \in \R'$ and all the reactants of $r$ are in $W$ then the product(s) of $r$ are also in $W$. The (F) condition now becomes the statement that
 each reactant of each reaction in $\R'$  is present in ${\rm cl}_{\R'}(F)$.   Note also that, assuming the (F) condition holds, the 
 (RA) condition becomes equivalent to the stronger condition that each reaction $r \in \R'$  is catalysed by at least one molecule type that is present in ${\rm cl}_{\R'}(F)$.
 
Two fundamental combinatorial results concerning RAFs (from   \citep{hor}) which will be applied in this paper are the following:
\begin{itemize}
\item If $\Q$ has a  RAF then it has a unique maximal RAF which contains all other RAFs for $\Q$ (referred to as the {\em maxRAF} of $\Q$, denoted ${\rm maxRAF}(\Q)$).
\item Determining whether or not $\Q$ has a RAF, and if so constructing ${\rm maxRAF}(\Q)$ can be solved by an algorithm that is polynomial-time in the size of $\Q$.
\end{itemize}

By contrast to the second point, finding a {\em smallest} RAF in a CRS $\Q$ has been shown to be NP-hard \cite{ste13}.  The maxRAF of the CRS shown in 
Fig.~\ref{fig1} is $\{r_1, r_2, r_3\}$.


\bigskip

{\bf Definitions:} (subRAFs, irrRAFs, closure, closed RAFs, CAFs)

We now introduce some further notions related to different types of RAFs.

The maxRAF of a CRS $\Q$ may contain  one or more proper subsets of reactions that are themselves RAFs for $\Q$, in which case we call any such subset a {\em subRAF} of the maxRAF. 

A RAF  $\R'$ is said to be an {\em irreducible} RAF (irrRAF) if it contains no proper subset of $\R'$ that is a RAF. 
In other words, removing any single reaction from an irrRAF  $\R'$ gives a set of reactions that does not contain a RAF for $\Q$.  Constructing an irrRAF for $\Q$ (or determining than none exists when $\Q$ has no RAFs) can also be carried out in polynomial-time \cite{hor}, however the number of irrRAFs can grow exponentially with the size of the CRS  \cite{hor12}.
To illustrate this notion, the RAF $\{r_2, r_3\}$ is the only irrRAF for the CRS in Fig.~\ref{fig1}.

Given any subset $\R'$ of reactions from $\R$, the {\em closure} of $\R'$ in $\Q$, denoted $\overline{\R'}$ is the (unique) minimal subset $\R''$  of $\R$ that contains $\R'$ and satisfies the property that if
a reaction $r$ from $\R$ has each of its reactants and at least one catalyst present in the food set or as a product of a reaction from $\R''$ then $r$ is in $\R''$.  
It is easily seen that the  closure of any RAF is always a RAF.  We say that a RAF $\R'$  is a {\em closed} RAF  if it is equal to its closure (i.e. $\R' = \overline{\R'}$). In particular, the maxRAF is always closed (closed RAFs are the type of RAF that is most closely related to, but still different from, organisations in Chemical Organisation Theory \citep{hor17a}).
Referring again to Fig.~\ref{fig1}, the closure of the RAF $\{r_2, r_3\}$ is the maxRAF $\{r_1, r_2, r_3\}$.

 A {\em minimal closed RAF} for a CRS $\Q$ is a closed RAF $\R'$ for $\Q$ that does not contain any other closed RAF for $\Q$ as a strict subset. Any closed irrRAF is a minimal closed RAF but a minimal closed RAF need not be an irrRAF.   Once again Fig.~\ref{fig1} illustrates this last concept: for this CRS, the minimal closed RAF is the maxRAF $\{r_1, r_2, r_3\}$ but it is not an irrRAF since it contains the RAF $\{r_2, r_3\}$.

Given a CRS $\Q=(X,\R, C, F)$, a stronger notion than a RAF is that of a {\em constructively autocatalytic F-generated} (CAF) set for $\Q$ (introduced in \cite{mos}).  A CAF for $\Q$ is a nonempty subset $\R'$ of $\R$ for which the reactions in $\R'$ can be ordered in such a way that for each reaction $r$ in $\R'$, each reactant and at least one catalyst of $r$ has the property that it is either produced by an earlier reaction from $\R'$ or is present in the food set. In other words, a CAF is like a RAF with the extra requirement that no spontaneous (uncatalysed) reactions are required for its formation (i.e. the catalyst needs to be already present when it is first needed). For example, both the RAFs in Fig.~\ref{fig1} fail to be a CAF.


\section{The structure of RAFs in `elementary' catalytic reactions systems}

Let CRS $\Q = (X, \R, C, F)$.  We say that $\Q$ is {\em elementary} if it satisfies the following condition:
\begin{itemize}
\item Each reaction $r$ in $\R$ has all its reactants in $F$.
\end{itemize}

An elementary CRS is a very special type of CRS; however it has arisen both in applications to real experimental chemical systems  \citep{ash, vai} and in theoretical models  \citep{jai98}.  The  CRS shown in Fig.~\ref{fig1} is not an elementary CRS, but it becomes so if reaction $r_5$ is removed.  
It is possible to extend the definition of elementary CRS  to also allow for reversible reactions, by requiring only one side of the reaction to contain molecule types that are exclusively from $F$.

In this section, we show that
elementary RAFs have sufficient structure to allow a very concise classification of their RAFs, closed subRAFs, irrRAFs, and `uninhibited' closed RAFs (a notion described below), something which is
problematic in general. We then extend this analysis to more complicated types of RAFs in the next section.  

Our analysis in this section relies heavily on some key notions from graph theory, so we begin by recalling some concepts from that area. 

\subsection{Definitions}
In this paper, all graphs will be finite.  Given a directed graph $\D=(V,A)$, recall that a {\em strongly connected component} of $\D$ is a maximal subset $W$ of $V$ with the property that for any vertices $u, v$ in $W$, there is a path from $u$ to $v$ and a path from $v$ to $u$.

It is a classical result that for any directed graph $\D=(V,A)$, the vertex set $V$ can be partitioned into strongly connected components. This, in turn, induces a directed graph structure, called the {\em condensation (digraph) of $\D$},  which we will denote by $\D^*$. In this directed graph, the vertex set is the collection of strongly connected components of $\D$ and there is an 
arc $(U,V)$ in $\D^*$ if there is an arc $(u,v)$ in $\D$ with $u \in U$ and $v \in V$.   By definition, $\D^*$ is an acyclic directed graph. 
Moreover, the task of partitioning $V$ into strongly connected components and constructing the graph $\D^*$ can both be carried out in polynomial time  \citep{tar}.
Note that the strongly connected component containing $v$ will consist just of $v$ if $v$ is not part of a cycle involving another vertex.

We now introduce some further definitions.   Given a directed graph $\D=(V,A)$: 

\begin{itemize}
\item We say that a strongly connected component $S$ of $\D$ is a {\em core} if either $|S|>1$ or $|S|=1$ (say $S=\{r\}$) and there is an arc from $r$ to itself. Note that $\D$ has a core if and only if $\D$ has a directed cycle.
\item A {\em chordless cycle} in a directed graph $\D= (V, A)$ is a subset $U$ of vertices of $\D$ for which the induced graph $\D|U$ is a directed cycle (here $\D|U = (U, A')$ where the arc set $A'$ for $\D|U$ is given by $A'= \{(u,v)\in A: u, v \in U\}$). Note that if $|U|=1$, this means that there is an arc from the vertex in $U$ to itself.
\item A vertex $v$ in $V$ is {\em reachable} from some subset $S$ of $V$ if there is a directed path from some vertex in $S$ to $v$. More generally, a subset $U$ of $V$ is reachable from $S$ if there is some vertex $v \in U$ that is reachable from $S$.
\end{itemize}
The terminology `core' follows a similar usage by \cite{vas}, in which the set of vertices (molecule types) that are reachable from a core is referred to as the `periphery' of the core. 

\subsection{First main result} 

The following theorem provides graph-theoretic characterisations of RAFs,  irrRAFs, closed RAFs, and minimal closed RAFs within any elementary CRS.

Given any CRS, $\Q$, consider the directed graph $\D_{\Q}$ with vertex set $\R$ and with an arc $(r, r')$ if a product of reaction $r$ is a catalyst of 
reaction $r'$.  In addition, for any reaction $r$ that has a catalyst in $F$, we add the arc $(r,r)$ (i.e. a loop) into $\D_{\Q}$ if this arc is not already present; this step is just a formal strategy to allow the 
results to be stated more succinctly, and does not necessarily mean that a product of $r$ is an actual catalyst of $r$.

\begin{thm}
\label{bulletthm}
Let $\Q$ be an elementary CRS.  Then:
\begin{itemize} 
\item[(i)] $\Q$ has a RAF if and only if $\D_{\Q}$ has a directed cycle, and this holds  if and only if $\D_{\Q}$ contains a chordless directed cycle.  The RAFs of $\Q$ correspond to the subsets $\R'$ of $\R$ for which the induced directed graph 
$\D_{\Q}|\R'$ has the property that each vertex has in-degree at least 1.
\item[(ii)] The irrRAFs of $\Q$ are the chordless cycles in $\D_{\Q}$. The closed irrRAFs of $\Q$ are chordless cycles from which no other vertex of $\D_{\Q}$ is reachable. 
The smallest RAFs of $\Q$ are  the shortest directed cycles in $\D_{\Q}$.
\item[(iii)] The closed RAFs of $\Q$ are the subsets of $\R$ obtained by taking the union  of any
 one or more cores of $\D_{\Q}$ and adding in all the reactions in $\R$  that are reachable from this union. 
\item[(iv)] Each minimal closed RAF of $\Q$ is obtained by taking any core $C$  of $\D_{\Q}$ for which no other core of $\D_{\Q}$ is reachable from $C$, and adding in all reactions in $\R$ that are reachable from $C$.
\item[(v)] The number of minimal closed RAFs of $\Q$ is at most the number of cores in $\D_{\Q}$, and thus it is bounded above by $|{\rm maxRAF}(\Q)|$. These can all be found and listed in polynomial time in $|\Q|$.
\item[(vi)] The question of whether or not a given RAF for $\Q$ (e.g. the maxRAF) contains a closed RAF as a strict subset can be solved in polynomial-time.
\end{itemize}
\end{thm}

\begin{proof}
\mbox{}
A key observation throughout is that in an elementary CRS $\Q$, any nonempty subset $\R'$ of reactions automatically satisfies the $F$--generated property,  so $\R'$ forms a RAF
for $\Q$ if and only if $\R'$ satisfies the reflexively autocatalytic (RA) property. By the manner in which  $\D_{\Q}$ is constructed, the RA property means that the induced subgraph  
$\D_{\Q}|\R'$ has the property that each vertex has in-degree at least 1.

In particular, $\R$ has a RAF if and only if $\D_{\Q}$  has a directed cycle. The `if' direction of this claim is clear.  For the `only if' direction, suppose that $\R'$ is a RAF and $r \in \R'$.  By the assumption that each vertex in $\D_{\Q}$ has in-degree at least 1, there is a directed walk of length $k$ (for any $k \geq 1$) involving vertices in  $\R'$ and  ending in $r$. Since $\R'$ is finite if we take $k> |\R'|$ then two vertices on this directed walk must coincide and the resulting sub-walk between this vertex to itself gives a directed cycle in $\D_{\Q}$. Moreover, $\D_{\Q}|\R'$ contains a directed cycle if and only if this sub-digraph contains a chordless cycle (again, the `if' direction is clear and the `only if' direction follows by the finiteness of $\R$,  so shortening each directed cycle by following a chord leads to a sequence of cycles of decreasing length that eventually terminates on a chordless cycle). 
This establishes  Part (i).

For Part (ii),  a subset  $\R'$ of $\R$ has the property that $\D_{\Q}|\R'$ is a chordless cycle, which implies (by Part (i)) that $\R'$ is a RAF. It is also an irrRAF;  otherwise, the cycle would
have a chord. Conversely, if $\R'$ has the property that $\D_{\Q}|\R'$ is not a chordless cycle, then either $\D_{\Q}$ does not contain a cycle (in which case it is not a RAF) or it contains a cycle
which either has a chord or has other vertices reachable from it, in which case $\R'$ is not an irrRAF. This establishes the first sentence of Part (ii).  The arguments for the second and third sentences follow similar lines. 

For Part (iii),  it is clear that the union of one or more cores is a RAF; however, the resulting set of reactions  is closed if and only if all reactions that are reachable from that set are also included.

For Part (iv), suppose that a core $c'$ is reachable from another core $c$ (by definition, $c$ is not reachable from $c'$). Any closed RAF $\R'$ that contains both $c$ and $c'$ is thus not minimal,  since we could delete $c'$ and all the reactions that are reachable from $c'$ but not from $c$ and obtain a strict subset of $\R'$ that is also a closed RAF. On the other hand, if $\R'$ has the property described in Part (iv), then it is a closed RAF by Part (iii) and it is also minimal, since any closed RAF must contain at least one core, alongside all the reactions that are reachable from it.

Part (v) follows from Part (iv), since each minimal closed RAF is associated with exactly one core, and since cores are strongly connected components of $\D_{\Q}$ these cores are vertex-disjoint (i.e. two cores share no reaction). Consequently, the number of cores is bounded above by the number of reactions in the maxRAF of $\D_{\Q}$.  Moreover,  finding the strongly connected components of any digraph can be done in polynomial time in the size of the digraph \citep{tar}. Each of these strongly connected components  can then be  tested in polynomial time  to determine if it is a core; if  so, one can then determine in polynomial time which other vertices  are reachable from it. Thus the minimal closed RAFs can be listed in polynomial time in the size of $\Q$.

Part (vi) follows from Part (v) since $\Q$ contains a closed subRAF if and only if it contains a minimal closed subRAF.

\end{proof}
\hfill$\Box$

Figs.~\ref{fig2}, \ref{fig3} and \ref{fig4} illustrate Parts (i)--(iv) of Theorem~\ref{bulletthm}.  Some of these examples are based on reaction networks that come from actual experimental RAF sets.
 \begin{figure}[htb]
\centering
\includegraphics[scale=1.0]{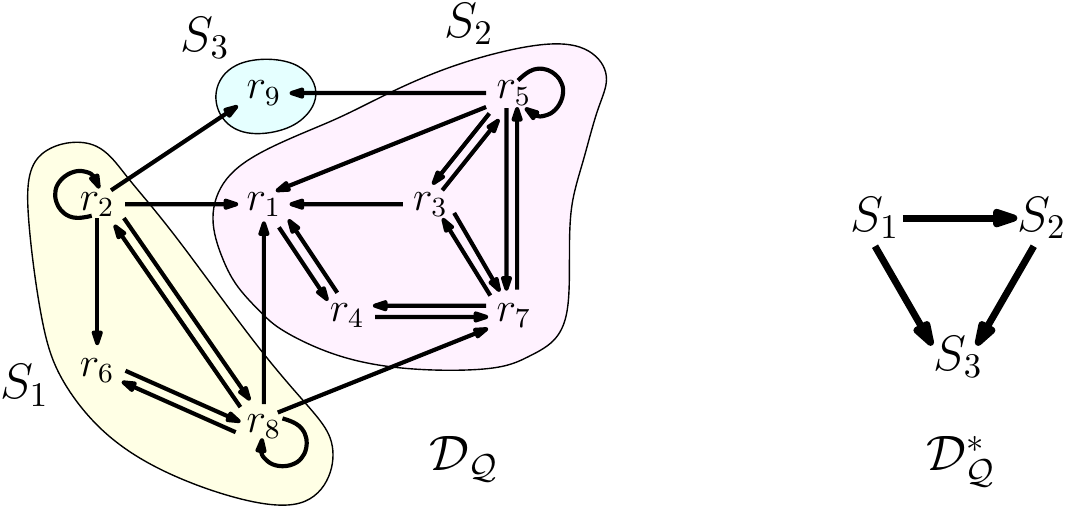}
\caption{The directed graph $\D_{\Q}$ for an elementary CRS  $\Q$ (adapted from an experimental system of \cite{ash})
that has three strongly connected components ($S_1, S_2, S_3$), of which $S_1$ and $S_2$ are cores. The associated (acyclic) condensation digraph $\D^*_{\Q}$ is shown on the right. 
  The unique minimal closed RAF is $S_2 \cup S_3$;  the other closed RAF is the full set itself, namely $S_1 \cup S_2 \cup S_3$. 
The reactions subsets $S_1$, $S_1 \cup S_3$, $S_1 \cup S_2$, $S_1 \cup S_3$, and $S_2$ are all RAFs but not closed RAFs.  A computer-based search finds 305 RAFs altogether. 
There are six chordless cycles in this CRS, which correspond to the six irrRAFs: $\{r_2\}, \{r_5\}, \{r_8\}, \{r_1, r_4\}, \{r_4, r_7\}$ and
$\{r_3, r_7\}$. 
Note that this representation of the CRS is in terms of the molecules produced by reactions that have reactants in the food set. However, each
reaction produces a single (and unique) product so we can identify the product with the reaction in this example.}
\label{fig2}
\end{figure}

 \begin{figure}[htb]
\centering
\includegraphics[scale=0.8]{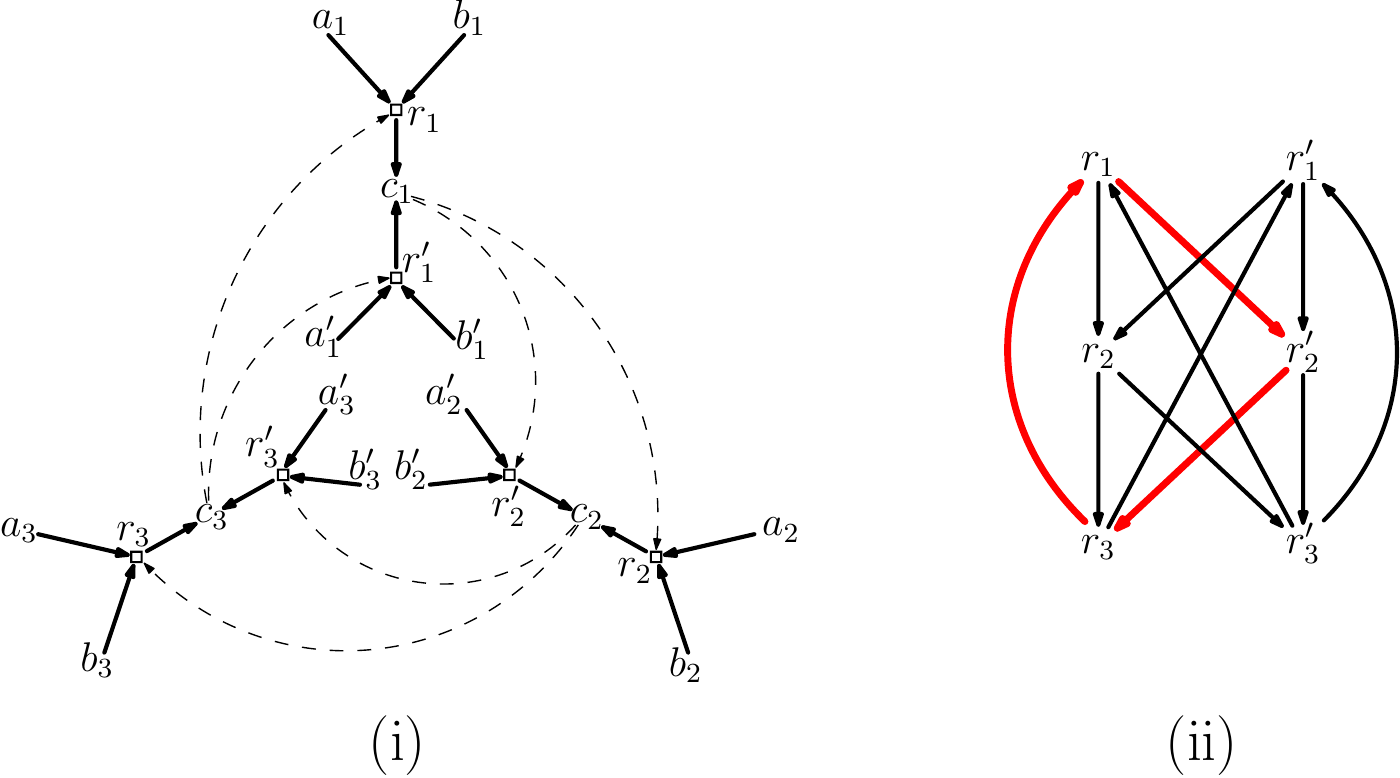}
\caption{(i)  An elementary CRS (with food set $F$ equal to the 12 elements labelled $a_i, a_i', b_i, b_i'$ for $i=1,2,3$) that has eight irrRAFs, each of which has size 3 (this example can be extended to  produce an elementary CRS with $2n$ reactions and $2^n$ irrRAFs \citep{hor12}).  These irrRAFs correspond to the eight chordless cycles in the graph $\D_{\Q}$ shown in (ii), with one of these chordless cycles indicated by the three bold arcs.  None of these irrRAFs is closed. There are 27 RAFs  for $\Q$ in total.}
\label{fig3}
\end{figure}

 \begin{figure}[htb]
\centering
\includegraphics[scale=1.2]{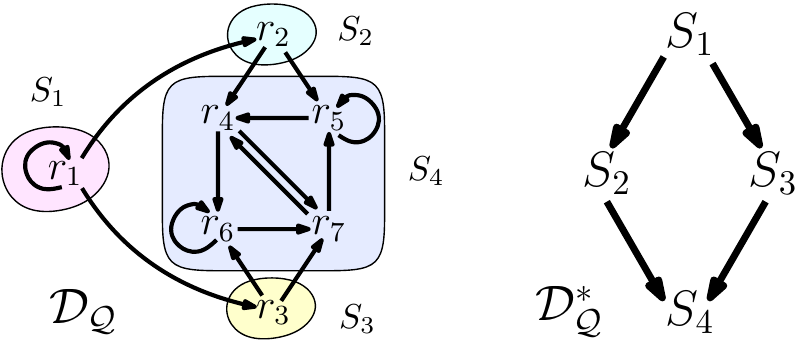}
\caption{The directed graph $\D_{\Q}$ for an elementary CRS (from an experimental system \citep{hor13}),
 shown on the left,  has 67 RAFs and two closed RAFs (the whole set and $\{r_4, r_5, r_6, r_7\}$). The strongly connected components of $\D_{\Q}$ are $\{r_1\}, \{r_2\}, \{r_3\}$ and $\{r_4, r_5, r_6, r_7\}$, two of which are cores (namely, $\{r_1\}$ and $\{r_4, r_5, r_6, r_7 \}$). The associated condensation digraph $\D^*_{\Q}$ is shown on the right. For this RAF, there are four irrRAFs, namely $\{r_1\}, \{r_5\}, \{r_6\}$, and $\{r_4, r_7\}$.}
\label{fig4}
\end{figure}

\newpage

{\bf Remarks:}  

\begin{itemize} 

\item
Parts (ii)--(vi) of  Theorem~\ref{bulletthm} hold even when $\Q$ is not elementary, provided that $\Q' = (X, \R', C, F)$ is elementary where
$\R'$ is the maxRAF of $\Q$.

\item 
Although cores share no reactions in common, it is quite possible for minimal closed RAFs to share
reactions in common. 

\item 
The last sentence of Part (ii) implies that the size of the smallest RAF is equal to the length of the shortest directed cycle in $\D_\Q$ and this can be found in polynomial time in $|\Q|$ (by a depth-first-search or network flow techniques).  This is in contrast to the problem of finding the size of a smallest RAF in a general CRS, which has been shown to be NP-hard in \cite{ste13}. 
 
\item 
An important extension of the RAF concept allows for molecule types to inhibit reactions (as well as being able to catalyse reactions).   For a general CRS $\Q$ it is known that determining whether or not a CRS  $\Q$ has a RAF $\R'$
for which no reaction is inhibited by any molecule produced by $\R'$ is NP-hard \cite{mos}. 
However, for any elementary CRS, Theorem~\ref{bulletthm}(v) provides the following positive result. 

\begin{cor}
When inhibition is also allowed in an elementary CRS $\Q$, it is possible to determine in polynomial time whether $\Q$ contains a closed RAF $\R'$ for which no reaction is inhibited by any molecule type produced by $\R'$.
\end{cor}
\begin{proof}
There is a closed RAF for $\Q$ that has no inhibition if and only if there is a minimal closed RAF for $\Q$ that has no inhibition. By Part (v) of Theorem~\ref{bulletthm}, there are at most $|{\rm maxRAF}(\Q)|$ minimal closed RAFs for an elementary CRS $\Q$, and these can all be checked in polynomial time to determine if any of them have the property that no reaction is inhibited by any molecule type produced by the reactions in the set.
\end{proof}

\item   Part (v) of Theorem~\ref{bulletthm} raises the question of whether this result might apply with the restriction that $\Q$ is elementary. In other words, is the number of minimal closed RAFs in a (general, nonelementary) CRS bounded polynomially in the size of $\Q$?   The answer turns out to be `no' as the following example shows.

Consider the CRS $\Q_k: = (X, \R, C, F)$ where
$$X= \{f, x, \theta\} \cup \{x_1, x'_1, \ldots, x_k, x'_k\} \cup \{\theta_1, \ldots, \theta_k\}, F=\{f\},$$ and  for  $[k] = \{1,2, \ldots, k\}$, the reaction set is:
$$\R = \{r_x, r_\theta\} \cup \{r_i: i \in [k]\} \cup \{r'_i, i \in [k]\}  \cup \{\overline{r_i}:i \in [k]\} \cup \{\overline{r'_i}: i \in [k]\},$$
where these reactions are described as follows (with catalysts indicated above the arrows):  
$$r_x: f \xrightarrow{\theta} x,$$
$$r_\theta: \theta_1+\theta_2 + \cdots + \theta_k \xrightarrow{\theta} \theta,$$
and for all $i \in [k]$:
$$r_i: x \xrightarrow{x_i} x_i, \mbox{ } r'_i: x \xrightarrow{x'_i} x'_i,$$
$$\overline{r_i}: x_i \xrightarrow{\theta_i} \theta_i, \mbox{ }  \overline{r'_i}: x'_i \xrightarrow{\theta_i} \theta_i.$$
Thus, $\Q_k$ has a food set of size 1, a reaction set of size $4k+2$, and $3k+3$ molecule types.
Fig.~\ref{fig5} provides a graphical representation of $\Q_3$.

\begin{figure}[htb]
\centering
\includegraphics[scale=0.8]{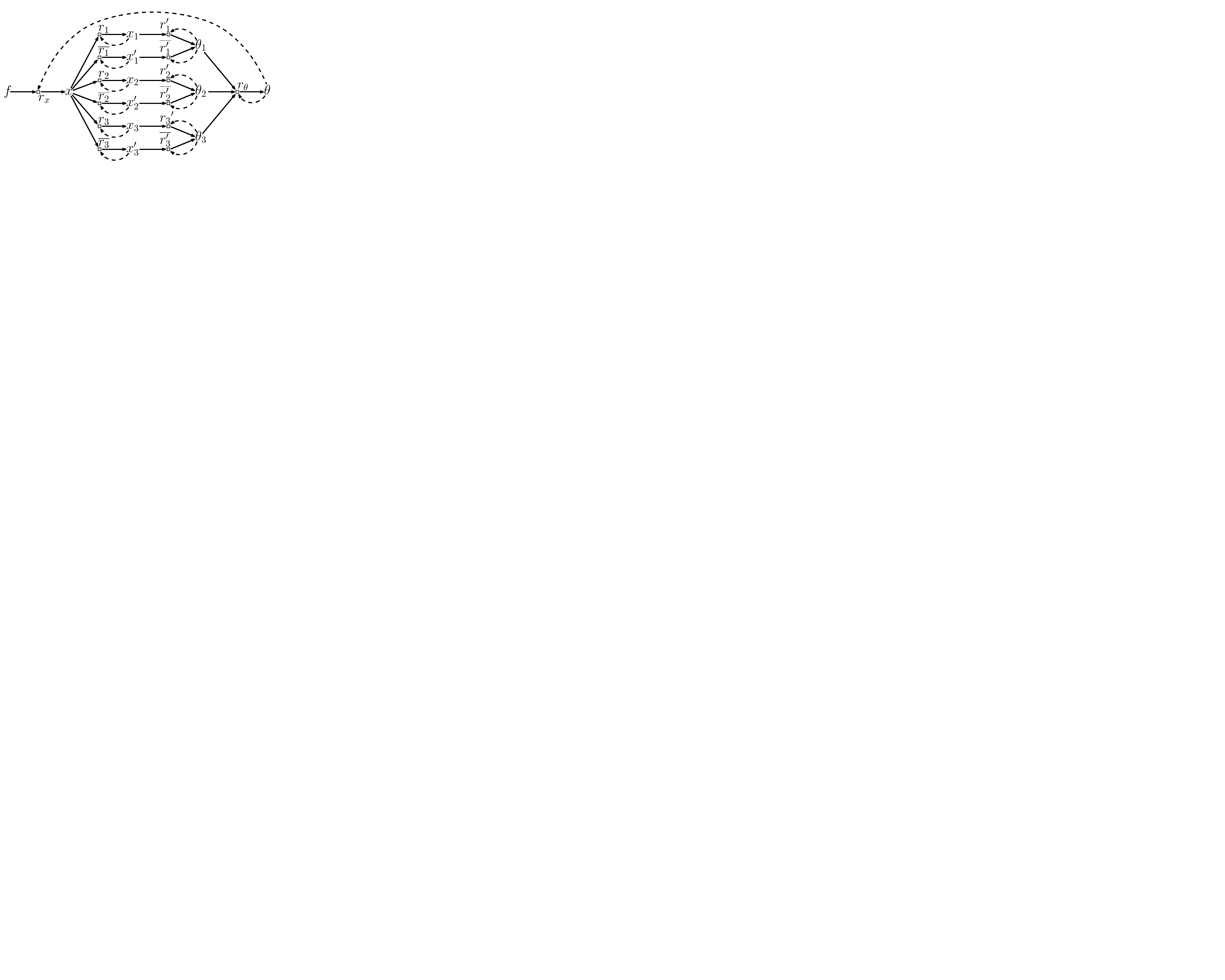}
\caption{The CRS $\Q_3$.}
\label{fig5}
\end{figure}

\begin{prop}
\label{thetapro}
The minimal closed RAFs of $\Q_k$ coincide with the irrRAFs for $\Q_k$, and there are $2^k$ of them. More precisely, $\R'$ is a minimal closed RAF of $\Q_k$ if and only if 
$\R'$ contains $r_x$ and  $r_\theta$ and for each $i \in [k]$, $\R'$  contains either
(i) $r_i$ and $\overline{r_i}$ but neither $r'_i$ nor $\overline{r'_i}$, or
(ii) $r'_i$ and $\overline{r'_i}$ but neither $r_i$ nor $\overline{r_i}$.
\end{prop}

\begin{proof}
The `if' direction in the second sentence is clear, since any such $\R'$ is easily seen to be a closed subRAF, as well as being an irrRAF, and thus is a minimal closed RAF.  For the `only if' direction, 
a subset $\R'$ of $\R$ is a RAF of $\Q_k$ precisely if the following two properties hold:  (a) $\R'$ contains $r_x$ and $r_\theta$, and (b) for each $i$, $\R'$ contains either $r_i$ and $\overline{r_i}$ or  $r'_i$ and $\overline{r'_i}$ (in order to generate $\theta_i$, which is required by $r_\theta$). Unless $\R'$ satisfies the stronger condition (i) or (ii)  (for each $i \in [k]$) listed in the statement of  Proposition~\ref{thetapro}, $\R'$ is not minimal.
\end{proof}
\end{itemize}

Another question that Part (v) of Theorem~\ref{bulletthm} suggests is the following: does an elementary CRS always have at most a polynomial number of closed RAFs?  Again, the answer is `no', and the construction to show this is much simpler than the previous example. 
Consider the elementary CRS with 
$F=\{f_1, \ldots, f_n\}, \mbox{ }  X= F \cup \{x_1, \ldots, x_n\},$
together with the set $\R$ of $k$ catalysed reactions
$r_i: f_i \xrightarrow{x_i} x_i$
for $i=1, \ldots, k$.  This CRS  has $2^k-1$ closed RAFs, one for each nonempty subset of $\R$.

\subsection{The probability of a RAF in an elementary CRS}

Given an elementary CRS $\Q$, suppose that catalysis is assigned randomly as follows: each molecule type catalyses each given reaction in $\R$ with a fixed probability $p$, independently across all pairs $(x,r)$  of 
molecule type $x$ and reaction $r$.  The  probability $p_\Q$ that $\Q$ has a RAF is simply the probability that $\D_{\Q}$ has a directed cycle (by Theorem~\ref{bulletthm}(i)).

In the case where  each reaction in $\R$ has just a single product, then the  asymptotic behaviour of $p_\Q$ as $|\R| \rightarrow \infty$ is equivalent to the emergence of a directed cycle in a large random directed graph, which has been previously studied in the random graph literature by \cite{bol}.

Here we provide a simple lower bound on $p_\Q$. Let $\lambda = p|\R|$ be the expected number of reactions that each molecule type catalyses.  
The following result gives a lower bound on $p_\Q$ that depends only on $\lambda$ and which converges towards 1 as $\lambda$ grows.
\begin{prop}
$p_\Q \geq 1-\left(1-\frac{\lambda}{|\R|}\right)^{|\R|} \sim 1- e^{-\lambda}$, where $\sim$ denotes asymptotic equality as $|\R|$ grows.
\end{prop}

\begin{proof}
Consider the probability $p_r$ that a single reaction $r$ has an arc to itself (such an event is sufficient but not necessary for $\D_{\Q}$ to contain a directed cycle).   If $r$ produces $m\geq 1$ products, we have $p_r = 1- (1-p)^m \geq pm \geq p = \lambda/|\R|.$
The probability that no reaction has an arc to itself is therefore $\left(1-\frac{\lambda}{|\R|}\right)^{|\R|}$. Since $(1-x/n)^n \sim e^{-x}$, we obtain the result claimed.
\end{proof}

\subsection{Eigenvector analysis}
A previous study by \cite{jai98} considered the dynamical aspects of an `autocatalytic set'  in a CRS, which is closely related to the notion of an RAF  (our graph $\D_\Q$ differs from theirs in two respects, firstly the vertices here represent reactions rather than molecule types, and we also permit self-loops from a reaction to itself). 
We now present the analogues of these earlier dynamical findings in our setting  (and formally, with proofs). 

Given an elementary CRS $\Q$, let $A_\Q$ denote the adjacency matrix of the directed graph $\D_{\Q}$.  Thus the rows and columns of $A_\Q$ are indexed by the reactions in 
$\R$ in some given order, and the entry of $A_Q$ corresponding to the pair $(r, r')$ is 1 precisely if $(r,r')$ is an arc of $\D_{\Q}$ and is zero otherwise.
By Perron-Frobenius theory for non-negative matrices, $A_Q$ has a non-negative real eigenvalue $\lambda$ of maximal modulus (amongst all the eigenvalues) and if
$\D_{\Q}$ is strongly-connected (i.e. $A_Q$ is irreducible), then $A_Q$ has a left (and a right) eigenvector with eigenvalue $\lambda$ whose components are all positive. 

The following results are analogues of the former study by \cite{jai98} to our setting.

\begin{prop}
\label{jaiprop}
\mbox{}
\begin{itemize}
\item[(i)]
If $\Q$ contains no RAF, then $\lambda=0$. 
\item[(ii)]
If $\Q$ contains a RAF, then $\lambda \geq 1$.
\item[(iii)]
If $A_Q$ has an eigenvalue $>0$  with an associated left eigenvector $w$, then the set of reactions $r$ for which $w_r>0$ forms a RAF for $\Q$.
\end{itemize}
\end{prop}

\begin{proof}
Part (i) follows from Part (i) of Theorem~\ref{bulletthm}, combined with the fact that the adjacency matrix $A$ of an acyclic directed graph is nilpotent (i.e. specifically, $A^{l+1}$ is the all-zero matrix when $l$ is the length of a longest path in the directed graph) and thus all the eigenvalues of $A$ are equal to zero \citep{cv}. 
For Part (ii), if $\Q$ contains a RAF, then $\D_{\Q}$ has a minimal (chordless) directed cycle (which could just be a loop on a vertex). Let $w$ be the vector that has value 1 for each vertex in this minimal directed cycle and is zero otherwise. Then $w$ is both a left and right eigenvector for $A_Q$ with eigenvalue 1. 
For Part (iii), let $\R'=\{r \in \R: w_r>0\}$.  The condition $wA_Q = \lambda w$ translates as $\sum_{r \in \R} w_{r} A_{rr'} = \lambda w_{r'}$. Since the right-hand side is non-zero for each reaction $r'\in \R'$, it follows that $w_{r}A_{rr'} \neq 0$ for at least one reaction $r \in \R'$; In other words, each reaction is $\R'$ is catalysed by the product of at least one reaction in $\R'$.
Since $\Q$ is elementary, this implies that $\R'$ is a RAF.
\end{proof}      

To illustrate an application of Proposition~\ref{jaiprop}, consider the system of 9 reactions from Fig.~\ref{fig2}.  In this case,
$\lambda \geq 1$ since the system contains a RAF ({\em cf.} Proposition~\ref{jaiprop}(ii)).  Regarding Part (iii),  three of the eigenvalues of $\A_\Q$ are strictly positive, and for the three corresponding left eigenvectors, one   has  strictly positive entries for the three reactions $r_2, r_6, r_8$, which  form the subRAF $S_1$ shown in Fig.~\ref{fig2}. A second left eigenvector has strictly positive entries for the reactions $r_1, r_3, r_4, r_5, r_7, r_9$, and these form the minimal closed subRAF $S_2 \cup S_3$ shown in Fig.~\ref{fig2}. The third left eigenvector has strictly positive entries for the reactions $r_1, r_4, r_7$ which forms
a subRAF of $S_2$. 

We end this section by noting that Part (iii) of Proposition~\ref{jaiprop} does not hold if left eigenvectors are substituted for right ones. A counterexample is given by the elementary CRS for which $A_\Q$ is the $2 \times 2$ matrix with both rows equal to $[0,1]$; in this case, $A_\Q$ has a principal eigenvalue  of $+1$ but the associated right eigenvector is a column vector with strictly positive entries, but this does not correspond to a RAF for $\Q$.

\section{Generative RAFs}

We now introduce a new notion which describes  how simple RAFs can develop into more complex ones in a progressive way.  This section will build on, and apply the results concerning elementary CRSs,  particularly Theorem~\ref{bulletthm}.

Given a CRS  $\Q = (X, \R, C, F)$ and a subset $Y$ of $X$ containing $F$, let $\R|Y$ be the subset of reactions in $\R$ that have all their reactants in $Y$, and let
$$\Q|Y:=(X, \R|Y, C, Y).$$
In other words,  $\Q|Y$ is the CRS obtained from $\Q$ by deleting each reaction  from $\R$  that does not have all its reactants in $Y$, and by expanding the food set to include all of $Y$.

\bigskip

{\bf Definition (genRAFs):} Given a CRS $\Q=(X,\R, C, F)$, we say that a RAF $\R'$ for $\Q$ is a  {\em genRAF} (or {\em generative RAF})  if there is a 
 sequence  $\R_1, \R_2, \ldots, \R_k$ of subsets of $\R$ with $\R_k = \R'$ and that satisfy the following properties:
\begin{itemize}
\item[(i)] $\R_{1}$ is the closure in $\Q$ of a RAF of $\Q|F$;
\item[(ii)] for each $i>1$, $\R_{i}$ is the closure in $\Q$ of a RAF of $\Q|Y_i$ where $Y_i=F \cup \pi(\R_{i-1})$, and where $\pi(\R_{i-1})$ refers to all molecule types
that are produced by a reaction from $\R_{i-1}$. 
\end{itemize}

\bigskip

Thus, a genRAF is any RAF for $\Q$ that can be formed by taking $\R_1$ to be the closure (within $\Q$)  of a RAF within the elementary CRS $\Q|F$, and for each $i>1$, adding the products of $\R_{i-1}$ to the food set $F$ of $\Q$
and taking $\R_i$ to be the closure (within $\Q$) of the  the resulting (induced) elementary CRS. In other words, the next closed RAF in the sequence is built upon an enlarged food set
generated by the previous closed RAFs in the sequence and considering just those reactions that use this enlarged food set as reactants, and then forming the closure of this set in $\Q$.

As an example, the CRS in Fig.~\ref{fig6}(a) is itself a genRAF (but not a CAF), as it has the generating sequence $\R_1, \R_2$ where $\R_1=\{r_1, r_2,\}$ and $\R_2 = \{r_1, r_2, r_3\}$,  however the CRS in Fig.~\ref{fig6}(b) is not a genRAF (even though it is an RAF).  

\begin{figure}[htb]
\centering
\includegraphics[scale=0.9]{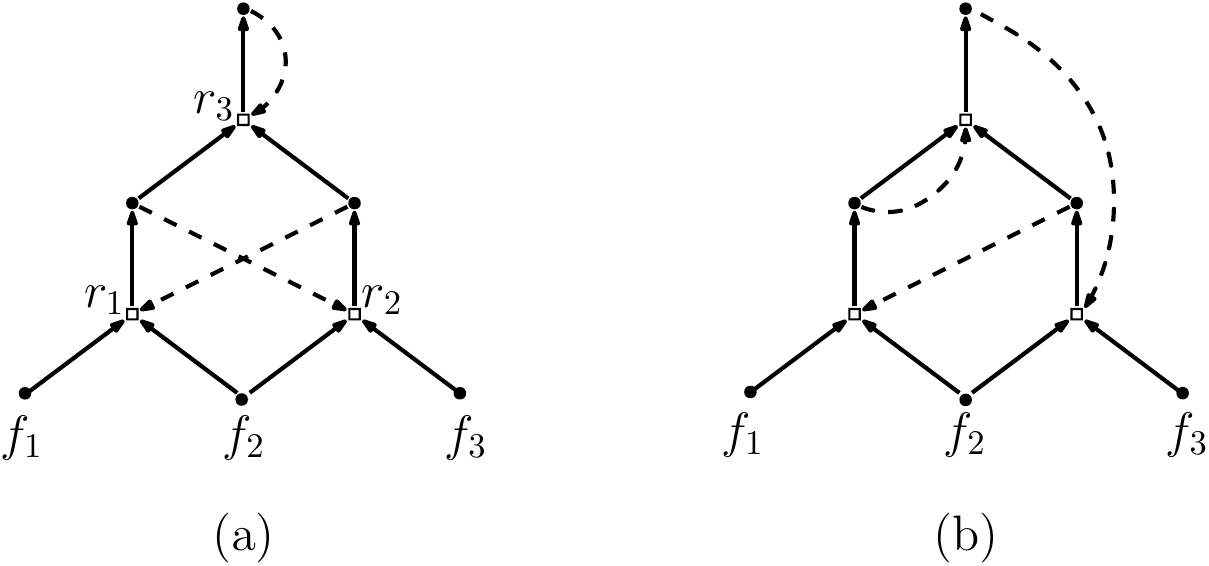}
\caption{(a): This CRS that is a genRAF (a generating sequence starts with the elementary closed RAF $\{r_1, r_2\}$,  and then adds $r_3$). (b): A different pattern of catalysis converts the three reactions into a RAF that is no longer a genRAF. In both cases, the food set is  $F=\{f_1, f_2, f_3\}$.}
\label{fig6}
\end{figure}

The motivation for considering the notion of  genRAFs is  two-fold.  Firstly a genRAF can be built up from simpler RAFs (starting with an elementary one)
by generating the required catalysts at each step (i.e. some reactions may still need to proceed initially uncatalysed, but a catalyst for the reaction will be generated by some other
reaction by the end of the same step).  This avoids the possibility of long chains of reactions that need to proceed uncatalysed until a catalyst for the very first link in the chain is produced, which seems less
biochemically plausible.  A second motivation for considering genRAFs is that they combine two further desirable properties: namely an emphasis on RAFs that are closed
(i.e. all reactions that are able to proceed and for which a catalyst is available will proceed), and genRAFs are sufficiently well-structured that some questions can be answered in polynomial time that are problematic for general RAFs (Theorem~\ref{secondmain}(iv) provides an explicit example). 

We will call the sequence $\R_1, \R_2, \ldots, \R_k$  in the above definition a {\em generating sequence} for $\R'$. 
We now make two observations, that are formalized in the following lemma.

\begin{lemma}
\label{nesting}
Suppose that a genRAF $\R'$ has generating sequence $\R_1, \R_2, \ldots, \R_k$. Then:
\begin{itemize}
\item[(i)] $\R'$ and each set in its generating sequence is a closed RAF for $\Q$. 
\item[(ii)]
$\R_i$ is a nested increasing sequence  (i.e.  $\R_i \subseteq \R_{i+1}$ for each $i \in \{1, \ldots, k-1\}$).  
\end{itemize}
\end{lemma}
\begin{proof}
Part (i) follows by definition, since each set in the generating sequence is the closure of a subRAF of $\Q$ and is therefore a closed RAF for $\Q$, and a genRAF is the final set in its generating sequence. 

Part (ii):  For each reaction $r \in \R$, let $\rho(r)$ denote the set of reactants of $r$.
We prove Part(ii) by induction on $k$.  For $k=2$, suppose that $r \in \R_1$.  Then $\rho(r) \subseteq F$ and there exists
some molecule type $x \in F \cup \pi(\R_1)$ that catalyses $r$.  Since $\R_2$ is a closed RAF, and since the reactants and at least one catalyst (namely $x$)
are available in the enlarged food set for $\R_2$, namely $Y_2 = F \cup \pi(\R_1)$,  then $r \in \R_2$. Thus Part (ii) holds for $k=2$.
Suppose now that Part (ii) holds for $k=m$ and that $\R_1, \R_2, \ldots, \R_{m+1}$ is a generating sequence for $\Q$.
We need to show that $\R_{m}\subseteq \R_{m+1}$.  To this end, suppose that $r \in \R_{m}$.  Then $\rho(r) \subseteq \pi(\R_{m-1})$ and there exists a molecule type
$x \in F \cup \pi(R_m)$ that catalyses $r$.  Now $\R_{m-1} \subseteq \R_m$ by induction and so $\rho(r) \subseteq \pi(\R_m)$. 
Thus the reactants and at least one catalyst of $r$ are in $Y_{m+1} =  F \cup \pi(R_m)$, and so, by the closure property, $r \in \R_{m+1}$. 
This establishes the induction step, and thereby Part (ii). 
\end{proof}

A natural question in the light of Lemma~\ref{nesting}(i) is the following: Is every closed RAF in a CRS generative?
 The answer to this is `no'  in general;  for example, a CRS may have a maxRAF
that requires too much `jumping ahead' with catalysis (chains of initially spontaneous reactions) to be built up in this way, as in  
Fig.~\ref{fig6}(b).  Shortly (Theorem~\ref{secondmain}) we will provide a precise, and efficiently checkable, characterisation for when a closed RAF is a genRAF.

Another instructive example is the following maxRAF that arose in a study of the binary polymer model from \cite{hor17}:
\begin{align*}
r_1 &: 10 + 0 \xrightarrow{01100} 100\\
r_2 &: 01+100 \xrightarrow{0} 01100\\
r_3 &: 10+1   \xrightarrow{0} 101\\
r_4 &: 11+10  \xrightarrow{101} 1110\\
r_5 &: 1110+0 \xrightarrow{101} 11100
\end{align*}
where $F=\{0, 1, 00, 01, 10, 11\}$. This maxRAF contains six subRAFs,  two of which are closed, namely, the full set  of all five reactions,  which is not generative, and the subset  $\{r_3, r_4, r_5\}$, which is a genRAF.

\bigskip

{\bf A maximal generative RAF:}
Given a CRS $\Q= (X, \R, C, F)$, consider the following sequence $(\overline{\R}_i, i \geq 1)$ of subsets of  $\R$.
Let $\Q_1 := \Q|F$,  let $\R_i ={\rm maxRAF}(\Q_1)$ and let $\overline{\R}_1$ be the closure of $\R_1$ in $\Q$.  
For $i>1$,  let $$\R_i = {\rm maxRAF}(\Q_i), \mbox{ where }  \Q_i:= F \cup \pi(\overline{\R}_{i-1}),$$
and let $\overline{\R}_i$ be the closure of $\R_i$ in $\Q$.

Note that $\R_1$ may be empty even if $\Q$ has a RAF (as Fig.~\ref{fig6}(b) shows), in which case, $\overline{\R}_i = \emptyset$ for all $i \geq 1$. 
However, if $\R_1$ is nonempty, then $\overline{\R}_i$ forms an increasing nested sequence of closed RAFs  for $\Q$ 
and so  the sequence  stabilises at some subset of reactions that we denote by 
 $\overline{\R}(\Q)$. Thus, $\overline{\R}(\Q) = \cup_{i\geq1} \overline{\R}_i$, and this set is identical to     
$\overline{\R}_k$ for some sufficiently large value of $k$ (with $k\leq |\R|$).

We can now state the main result of this section.

\begin{thm}
\label{secondmain}
Suppose that $\Q= (X, \R, C, F)$ is a CRS.
\begin{itemize}
\item[(i)]
$\Q$ contains a genRAF if and only if $\R_1 \neq \emptyset$, in which case
 $\overline{\R}(\Q)$ is a genRAF for $\Q$ that contains all other genRAFs for $\Q$.
 \item[(ii)]
 If $\R'$ is a closed RAF for $\Q$ then $\R'$ is a genRAF for $\Q$ if and only if $\R' =\overline{\R}(\Q')$,  where $\Q' = (X, \R', C, F)$.
\item[(iii)]
The construction of $\overline{\R}(\Q)$ and determining whether an arbitrary closed RAF $\R'$ for $\Q$ is generative  can be determined in polynomial time in $|\Q|$.
\item[(iv)]
Determining whether a given closed genRAF $\R'$ contains a  strict subset that is a closed RAF for $\Q$ can be solved in polynomial time in  $|\Q|$.
 \end{itemize}
\end{thm}

\begin{proof}
For the first claim in  Part (i), if $\R_1=\emptyset$ then $\Q|F$ contains no RAF and so $\Q$ has no genRAF.  
Suppose that $\R_1\neq \emptyset$.  Then $\overline{\R}(\Q)$ is a genRAF for $\Q$ since it has the generating sequence $\overline{\R}_i$ ($i \geq 1$) (noting that  $\overline{\R}_i$ is the closure in $\Q$ of $\R_i$ which is the maxRAF (and so a RAF) for $\Q|F$ when $i=1$, and for $\Q|(F \cup \pi(\overline{\R}_{i-1})$, when $i>1$).  
For the second claim in Part (i),  suppose that  $\R'$ is a genRAF for $\Q$; will show that $\R' \subseteq \overline{\R}(\Q)$. Let $(\R'_i, i \geq 1)$ be a generating sequence for $\R'$.  We show by induction on $i$ that $\R'_i \subseteq \overline{\R}_i$ for all $i>1$.  The base case $i=1$ holds since $\R_1$ is the maxRAF of $\Q|F$ which contains any other RAF of $\Q|F$, and so the closure of $\R_1$ in $\Q$, namely $\overline{\R}_1$  contains the closure in $\Q$ of any other RAF of $\Q|F$.
Suppose the induction hypothesis holds for  all values of $i$ up to $j\geq 1$. Then $\R_{j+1}$ is the maxRAF of $\Q|(F \cup \pi(\overline{\R}_j))$ and so it contains any RAF of $\Q|(F \cup \pi(\R'_j))$ since
$\R'_j \subseteq \overline{\R}_j$ (by the induction hypothesis) and so $F \cup \pi(\R'_j) \subseteq F \cup \pi(\overline{\R}_j)$.  Consequently,  the closure of $\R_{j+1}$ in $\Q$, namely, $\overline{\R}_{j+1}$,  contains the closure in $\Q$ of any RAF of $\Q|(F \cup \pi(\R'_j))$.  Thus the induction hypothesis holds, which establishes that $\R' \subseteq \overline{\R}(\Q)$.

For Part (ii), observe that $\overline{\R}(\Q')$ is a genRAF for $\Q'$ (by Part (i)) and so if  $\R' =  \overline{\R}(\Q')$ then $\R'$ is a genRAF for $\Q'$. Since $\R'$ is a closed RAF for $\Q$, $\R'$ is also a genRAF for $\Q$ (since the closure  in $\Q$ of any subset of reactions from $\R'$ is a subset of $\R'$). 
Conversely, suppose that $\R'$ is a genRAF for $\Q$. Then since $\R'$ is a closed RAF for $\Q$,  $\R'$ is also a genRAF for $\Q'$.  Now,  $\overline{\R}(\Q') \subseteq \R'$, and since  $\overline{\R}(\Q')$ contains any other genRAF for $\Q'$ (in particular, $\R'$) by Part (i),  we have $\overline{\R}(\Q')  =  \R'$, as required.

For Part (iii), the proof of the claim (regarding the construction of $\overline{\R}(\Q)$)  follows from the fact that the maxRAF (of $\Q_i$), and its closure (in $\Q$) can be computed in 
polynomial time in the size of the CRS \citep{hor}. The  the second claim then follows from Part (ii).

For Part (iv), consider the following algorithm.  
Given a closed RAF $\R'$ for $\Q$, let $\overline{\R}'_1, \overline{\R}'_2, \ldots$  be the generating sequence for $\overline{\R}(\Q')$ (described above, but with $\R$ replaced by $\R'$ and $\Q$ by $\Q' = (X, \R', C, F)$). From Part (ii) we have $\overline{\R}(\Q')= \R'$, and so $\overline{\R}'_1, \overline{\R}'_2, \ldots$ is a generating sequence for $\R'$.

Now, let  $\Q'_1 = \Q'|F$ and for $i>1$, let  $\Q'_i  = \Q|(F \cup \pi(\overline{\R}'_{i-1})).$  Notice that $\Q'_1$ is an elementary CRS, and for $i>1$ we can regard $\Q'_i$ as an elementary RAF with enlarged food set $F \cup \pi(\overline{\R}'_{i-1})$.
Thus we can apply Part (v) of Theorem~\ref{bulletthm}, and in polynomial time in $|\Q|$ search all the minimal closed RAFs for $\Q_j$ and determine whether the closure in $\Q$ of any of these results in a strict subset (say $\R''$) of $\R'$. 
 When such a set $\R''$ exists, its closure is clearly a closed RAF for $\Q$ that is a strict subset of $\R'$. However, if no such 
set $\R''$ is located then we claim that $\R'$ contains no closed RAF for $\Q$ as a strict subset. To see why, suppose that there is a closed RAF for $\Q$ that is strictly contained within
$\R'$. In that case there exists a minimal closed RAF for $\Q$ that is strictly contained in $\R'$, and we denote such a minimal closed RAF as $\R_*$. Let $j$ be the smallest value of 
$i$ for which $\R_*$ is contained in $\overline{\R}'_i$ as a strict subset (this is well defined, since $\R_*$ is strictly contained in $\R'$).   Then $\R_*$  is a  closed RAF for $\Q_j$ also, and its closure in $\Q$ is a strict subset of $\R'$, so the closure in $\Q$ of any minimal closed RAF for $\Q_j$ that lies strictly within $\R_*$ would also be a strict subset of $\R'$.

\end{proof}
\hfill$\Box$

{\bf Remarks:}  

\begin{itemize}
\item
If a CRS has a CAF (defined at the end of the Introduction), then the (unique) maximal CAF is generative.  However, a genRAF need not necessarily correspond to a maximal CAF.

\item 
Part (iv) of Theorem~\ref{secondmain} provides an interesting contrast to the general RAF setting. There the question of determining whether a closed RAF (e.g. the maxRAF) in an arbitrary CRS contains another closed RAF as a strict subset has unknown complexity.


\end{itemize}

\bigskip

\section{RAFs with reaction rates}

In this section, we consider a further refinement of RAF theory, by explicitly incorporating reaction rates into the analysis. This conveniently addresses one shortcoming implicit in the generative RAF definition from the last section -- namely
a generative RAF necessarily grows as a monotonically increasing nested system with the length of its associative generating sequence (Lemma~\ref{nesting}).   However, once a sufficiently large generative RAF is established,  one or more of its subRAFs  may then become dynamically favoured  if it is more `efficient' (i.e. all its reactions proceed at higher reaction rates than the generative RAF it lies within), as we shortly illustrate with a simple example.

Suppose that we have a CRS $\Q = (X, \R, C, F)$ and a function $f: C \rightarrow \RR^{\geq 0}$ that assigns a non-negative real number to each pair $(x, r) \in C$.
The interpretation here is that  $f(x, r)$ describes the {\em rate} at which reaction $r$ proceeds when the catalyst  $x$ is present.  

Given $\Q$ and $f$, together with a RAF $\R'$ for $\Q$, let:
$$\varphi(\R') =\min_{r \in \R'} \{\max\{f(x,r): (x, r) \in C, x \in {\rm cl}_{\R'}(F)\}\}$$
In other words,  $\varphi(\R')$ is the rate of the slowest reaction in the RAF $\R'$ under the most optimal choice of catalyst for each reaction in $\R'$ amongst those catalysts that are present in ${\rm cl}_{\R'}(F)$.
 
\begin{figure}[htb]
\centering
\includegraphics[scale=1.2]{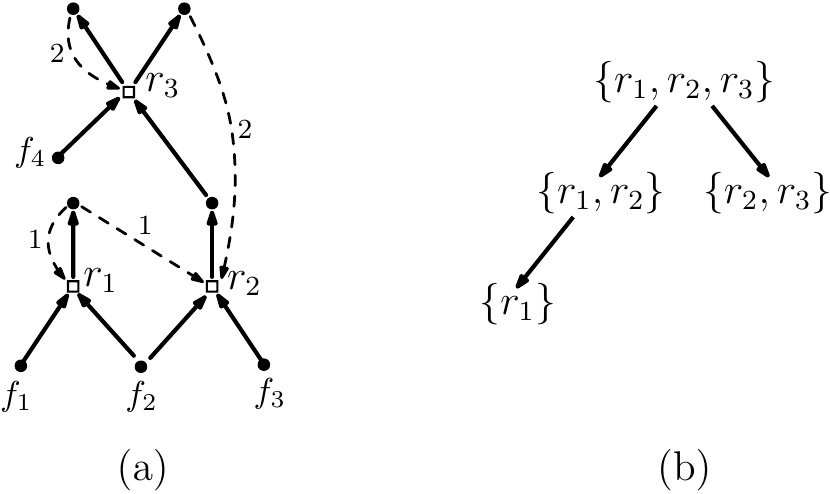}
\caption{(a) A RAF in which the catalysis arcs have associated rates (namely, the values 1 and 2 as indicated).  The poset consisting of  the maxRAF and its three subRAFs (partially ordered by set inclusion)  is shown by the Hasse diagram in (b).  All four RAFs have $\varphi$--values of $1$ except for the subRAF $\{r_2, r_3\}$, which has a $\varphi$--value of $2$. This optimal RAF $\{r_2, r_3\}$ is not a generative RAF (whereas the other three RAFs are generative; indeed,  $\{r_1\}$ and $\{r_1, r_2\}$ are elementary).  Nevertheless, once the generative maxRAF $\{r_1, r_2, r_3\}$ has formed,  $\{r_2, r_3\}$ can then emerge as the dominant sub-RAF.}
\label{fig7}
\end{figure}

{\bf Example:}  Fig.~\ref{fig7} provides an example to illustrate the notions above.  In this CRS the three reactions comprises a RAF, with a $\varphi$--value equal to 1.  However there are three subRAFs, and one of these (namely $\{r_2, r_3\}$) has a higher $\varphi$--value.  However, the less optimal closed subRAF  $\{r_1, r_2\}$ is generative and likely to have formed before the optimal one;  otherwise $\{r_2, r_3\}$ would require a chain of two reactions to occur uncatalysed  ($r_2$ followed by $r_3$)  before the catalysts for them become available. The closed RAF $\{r_1, r_2\}$ may then expand to $\{r_1,r_2, r_3\}$ before 
 this second closed RAF is subsequently out-competed by its subRAF $\{r_2, r_3\}$, since the catalysed reactions in this subRAF run twice as fast as the reaction $r_1$.

Our main result in this section shows that finding a RAF to maximise $\varphi$ can be achieved by an algorithm that runs in polynomial time in the size of $\Q$.  

\begin{thm}
\label{thm3}
\mbox{}
 There is a polynomial-time algorithm to construct a RAF with largest possible $\varphi$--value from any CRS $\Q$ that contains  RAF.  Moreover, this constructed RAF is the maximal RAF with this $\varphi$--value.
\end{thm}

\begin{proof}
Let $\cL = \{f(x, r): (x,r) \in C\}$, and let $M = \max \cL$. 
Consider the  CRS $\Q' = (X', \R^*, C^*, F)$ obtained from $\Q$ by first deleting any uncatalysed reaction and then replacing each reaction $r$ that is catalysed by (say) $k \geq 1$ molecule types
with $k$ distinct copies of this reaction $r_1, \ldots, r_k$, each of which is catalysed by a different one of the $k$ molecule types. 
Thus each reaction $r$ in $\R^*$ is catalysed by exactly one molecule type, which we will denote
as $x(r)$.  For the associated catalysis set  $C^* = \{(x(r), r): r \in \R^*\}$,  let
$f'$ be the rate function induced by $f$ (i.e. if  $r \in \R$ is replaced by $r_1, \ldots, r_k \in \R^*$ then 
$f'(x(r_i), r_i) := f(x(r), r)$).
For each $\ell \in \cL$ let: 
$$\R^*_\ell = \{r \in \R^*: f'(x(r), r) \geq \ell \}.$$
In other words,  $\R^*_\ell$ is the set of catalyst-reaction pairs $(x(r), r)$ where the rate of reaction $r$ when catalysed by the molecule type $x(r)$ is at least $\ell$ (as specified by the rate function $f$).

Now,  let $\widetilde{\R}$  be the maxRAF of $\R^*_\ell$ for the largest value of $\ell \in \cL$ for which maxRAF($\R^*_\ell$) is nonempty.
This set is well-defined, since  $\R^*= \R^*_\ell$ when $\ell = 0$, and because $\R$ (and thereby $\R^*$) is  assumed to have a RAF. 
Notice that $\widetilde{\R}$ can be efficiently determined, by starting at $\ell = M$ and decreasing $\ell$ through the (at most 
$|\cL|\leq |C|$) possible values it can take until a nonempty maxRAF first appears (alternatively, one could start at $\ell=0$ and increase $\ell$ until the last value for which a nonempty maxRAF is present).

{\em Claim}: $\widetilde{\R}$ is a RAF that has the largest possible $\varphi$--value of any RAF for  $\Q'$, and $\widetilde{\R}$ contains any other RAF for $\Q'$ with this maximal $\varphi$--value.

To establish this claim,  suppose that $\widetilde{\R} = {\rm maxRAF}(\R_\ell)$ for $\ell = t$ and that ${\rm maxRAF}(\R_\ell) = \emptyset$ for
$\ell > t$ (i.e. $t$ is the largest value of $\ell$ in $\cL$ for which $\R_\ell$ has a (nonempty) maxRAF).
For each reaction $r$ in   $\widetilde{\R}$, we then have $f'(x(r), r) \geq t$, and for at least one reaction  $r$ in
  $\widetilde{\R}, f'(x(r), r) = t$ (otherwise, a larger value of $\ell$ would support a maxRAF). It follows that  $\varphi(\widetilde{\R})=t$.
Now if $\R'$ is any other RAF for $\Q'$, let $t'$ be the minimal value of $f'(x(r), r)$ over all choices of $r \in \R'$.
Then $t' \leq t$ otherwise, $\R_\ell$ would have a nonempty maxRAF for a value $\ell =t'$ that is greater than $t$, contradicting the maximality of $t$.  Thus $\R' \subseteq \R^*_t$ and so $$\R' = {\rm maxRAF}(\R') \subseteq {\rm maxRAF}(\R^*_t) =\widetilde{\R},$$ which shows
that $\widetilde{\R}$ contains any other RAF with this maximal value. 

This establishes the above Claim, and thereby Theorem~\ref{thm3} for $\Q'$. However, the subset of  reactions of $\R$ 
whose copies  are present in $\widetilde{\R}$ provides a RAF for $\Q$ that has the largest
possible $\varphi$--value (namely $t$) and which contains any other RAF for $\Q$ with this $\varphi$--value.
\hfill$\Box$

\end{proof}

{\bf Remark:}  For the example in Fig.~\ref{fig7}, we have the subRAFs 
$\R_1 = \{r_1\}, \R_2 = \{r_2, r_3\}$ with $\varphi(\R_1)< \varphi(\R_2)$.
In this case, there is a path in the poset from $\R_1$ to $\R_2$ on which $\varphi$ is non-decreasing 
(this path goes `up' then `down' in Fig.~\ref{fig7}(b)). An interesting question might be to determine when this holds: in other words, from a sub-optimal RAF, can a more optimal RAF be reached by a chain of RAFs that, at each stage, either adds certain reactions or deletes one or more reactions, and so that the optimality score (as measured by $\varphi$) does not decrease?

\subsection{Rates for `catalytic ensembles'}

We can extend the results on rates in the previous section to accommodate the following feature: a reaction for which a combination of two or more catalysts is present may proceed at a rate that is higher than if just one  catalyst is present.  

We formalize this as follows.  Recall that in a  CRS $\Q = (X, \R, C, F)$, the set $C$  represents the pattern of catalysis and is a subset of $X \times \R$. Thus $(x,r) \in C$ means that $x$ catalyses reaction $r$.   Now suppose we wish to allow a combination (ensemble)  of one or more molecules to act as a catalyst for a reaction. In this case,  we can represent the CRS as a quadruple
 $\Q = (X, \R, \C, F)$ where $\C \subseteq (2^X-\emptyset)  \times \R$ and where $(A,r) \in \C$ means that the ensemble of molecules in $A$ acts as a (collective) catalyst for $r$, provided they are all present. 
 We refer to $\Q$ as a {\em generalised} CRS.
The notions of RAF, subRAF, CAF, and so on, can be generalized naturally. For example,  the RA condition for a subset $\R'$  is that  for each reaction $r$, there is a pair $(A, r) \in \C$ where each of the molecule types in $A$ is in the closure of $F$ relative to $\R'$.

Note that an ordinary CRS can be viewed as a special case of a generalised CRS by identifying $(x,r)$ with the pair $(\{x\}, r)$.
Note also that each reaction may have several ensembles of possible catalysts, and some (or all of these) may be just single molecule types.

Given a generalised CRS  $\Q = (X, \R, \C, F)$ we can associate an ordinary CRS $\Q' =(X', \R', C', F)$ to $\Q$ as follows.  
Let $$\A_\C := \{A \subseteq 2^X - \emptyset: \exists r \in \R: (A,r) \in \C\};$$
(so $\A_\C$ is the collection of catalyst ensembles in $\Q$).
For each $A \in \A_C$, let $x_A$ be a new molecule type, and let $r_A$ be the (formal) reaction
$A \rightarrow x_A$. 
Now let
\begin{align*}
 X' := & X\mbox{ } \dot\cup \mbox{ } \{x_A: A \in \A_\C\};  \\
\R' := & \R \mbox{ } \dot\cup \mbox{ } \{r_A: A \in \A_\C\}; \mbox{ and }\\
C' := &  \{(x_A, r): (A, r) \in \C\} \mbox{ } \dot\cup \mbox{ } \{(x_A, r_A): A \in \A_C\}.
\end{align*}
Note that $C' \subseteq X' \times \R'$.

In other words, $\Q'$ is obtained from $\Q$ by replacing each catalytic ensemble $A$ by a new molecule type $x_A$ and adding in the reaction  $r_A: A \rightarrow x_A$ catalysed by $x_A$.
The proof of the following lemma is straightforward.

\begin{lemma}
\label{helps}
A generalised CRS $\Q$ has a RAF if and only if the associated ordinary CRS $\Q'$ has a RAF that contains at least one reaction from $\R$. 
Moreover, in this case, the RAFs of $\Q$ correspond to the nonempty intersections of RAFs of $\Q'$ with $\R$.
\end{lemma}

Now suppose that we have a generalised CRS $\Q = (X, \R, \C, F)$ and a function $f: \C \rightarrow \RR^{\geq 0}$. The interpretation here is that  $f(A, r)$ describes the {\em rate} at which reaction $r$ proceeds when the catalyst ensemble $A$ is present.  

Given a RAF $\R'$ for $\Q$, let:
$$\varphi(\R') :=\min_{r \in \R'} \{\max\{f(A,r): (A, r) \in \C, A \subseteq {\rm cl}_{\R'}(F)\}\}$$
In other words,  $\varphi(\R')$ is the rate of the slowest reaction in the RAF $\R'$ under the most optimal choice of catalyst ensemble for each reaction in $\R'$ amongst catalyst ensembles that are subsets of ${\rm cl}_{\R'}(F)$.

Lemma~\ref{helps} now provides the following corollary of Theorem~\ref{thm3}.

\begin{cor}
\label{cor}
There is a polynomial-time algorithm to construct a RAF for $\Q$ with  largest possible $\varphi$--value from any CRS $\Q$ that contains a RAF.   Moreover, this constructed RAF is the maximal RAF for $\Q$  with this $\varphi$--value.
\end{cor}

\section{Concluding comments}

In this paper, we have considered special types of RAFs that allow for exact yet tractable mathematical and algorithmic analysis, and which also incorporate additional biochemical realism (restricting the depth of uncatalysed reactions chains in generative RAFs and 
allowing reaction rates).  

We first considered the special  setting of `elementary' systems in which all reactions (or at least those present in the maxRAF) have all their reactants present in the food set. This allows for the structure of the collection of RAFs, irrRAFs, and closed subRAFs to be explicitly described graph-theoretically. As a result, some problems that are computationally intractable in the general CRS setting turn out to be polynomial-time for an elementary CRS. For example, one  can efficiently find the smallest RAFs in an elementary CRS, which is an NP-hard problem in general  \citep{ste13}.
Also, the number of minimal closed subRAF in an elementary CRS is linear in the size of the set of reactions (for a general CRS,  they can be exponential in number).  For future work, 
it may  be of interest to determine if there are polynomial-time algorithms that can answer the following questions for an elementary CRS:  (i) What is the size of the largest irrRAF? (ii) If inhibition is allowed, then is there a RAF that has no inhibition? 

The  concept of an `elementary' CRS is an all-or-nothing notion. One way to extend  the results above could be to define  the notion of `level', whereby a CRS has level $k$ if
the length of the longest path from the food set to any reaction product goes through at most $k$ reactions (an elementary CRS thus has level 1).  We have not explored this further here but instead, we consider the related alternative notion of a generative RAF. Briefly, a generative RAF allows a RAF to form by effectively enlarging its `food set' with products of reactions,  so that each step only requires catalysts that are either present or produced by reactions in the RAF at that stage.
Although generative RAFs are more complex than elementary ones, their close connection to elementary RAFs
(in a stratified way) allows for a more tractable analysis than for general RAFs.  Moreover, unlike elementary RAFs, no special assumption is required on the underlying CRS; generative RAFs are just a special type of RAF that can be generated in a certain sequential fashion in any CRS.

In the final section, we considered the impact of rates of RAFs (which need not be generative), and particularly the algorithmic question of finding a RAF that maximises the rates of its slowest reaction.
Not only is this problem solvable in the size of the CRS, but it can also be extended to the slightly more general setting of allowing `catalytic ensembles'.  The introduction of rates allows for the study of how  a population of different closed subRAFs might evolve over time,  in which primitive subRAF are replaced (out-competed) by efficient ones that rely on new catalysts in place of more primitive ones. We hope to explore this further in future work.

\section*{Acknowledgments}

WH thanks the {\em Institute for Advanced Study}, Amsterdam, for financial support in the form of a fellowship.

\bibliographystyle{siamplain}
\bibliography{RAF_rates}

\begin{mcitethebibliography}{19}
\providecommand*\natexlab[1]{#1}
\providecommand*\mciteSetBstSublistMode[1]{}
\providecommand*\mciteSetBstMaxWidthForm[2]{}
\providecommand*\mciteBstWouldAddEndPuncttrue
  {\def\EndOfBibitem{\unskip.}}
\providecommand*\mciteBstWouldAddEndPunctfalse
  {\let\EndOfBibitem\relax}
\providecommand*\mciteSetBstMidEndSepPunct[3]{}
\providecommand*\mciteSetBstSublistLabelBeginEnd[3]{}
\providecommand*\EndOfBibitem{}
\mciteSetBstSublistMode{f}
\mciteSetBstMaxWidthForm{subitem}{(\alph{mcitesubitemcount})}
\mciteSetBstSublistLabelBeginEnd
  {\mcitemaxwidthsubitemform\space}
  {\relax}
  {\relax}

\bibitem[Kauffman(1986)]{kau1}
Kauffman,~S.~A. \emph{J. Theor. Biol.} \textbf{1986}, \emph{19}, 1--24\relax
\mciteBstWouldAddEndPuncttrue
\mciteSetBstMidEndSepPunct{\mcitedefaultmidpunct}
{\mcitedefaultendpunct}{\mcitedefaultseppunct}\relax
\EndOfBibitem
\bibitem[Kauffman(1993)]{kau2}
Kauffman,~S.~A. \emph{The Origins of Order}; Oxford University Press,
  1993\relax
\mciteBstWouldAddEndPuncttrue
\mciteSetBstMidEndSepPunct{\mcitedefaultmidpunct}
{\mcitedefaultendpunct}{\mcitedefaultseppunct}\relax
\EndOfBibitem
\bibitem[Hordijk and Steel(2017)Hordijk, and Steel]{hor17}
Hordijk,~W.; Steel,~M. \emph{Biosystems} \textbf{2017}, \emph{152}, 1--10\relax
\mciteBstWouldAddEndPuncttrue
\mciteSetBstMidEndSepPunct{\mcitedefaultmidpunct}
{\mcitedefaultendpunct}{\mcitedefaultseppunct}\relax
\EndOfBibitem
\bibitem[Jaramillo \latin{et~al.}(2010)Jaramillo, Honorato-Zimmer, Pereira,
  Contreras, Reynaert, Hern{\'a}ndez, Soto-Andrade, C{\'a}rdenas,
  Cornish-Bowden, and Letelier]{jar}
Jaramillo,~S.; Honorato-Zimmer,~R.; Pereira,~U.; Contreras,~D.; Reynaert,~B.;
  Hern{\'a}ndez,~V.; Soto-Andrade,~J.; C{\'a}rdenas,~M.~L.; Cornish-Bowden,~A.;
  Letelier,~J.~C. (M,R) systems and RAF sets: common ideas, tools and
  projections. Proceedings of the ALIFEXII Conference, Odense Denmark, August
  2010. 2010; pp 94--100\relax
\mciteBstWouldAddEndPuncttrue
\mciteSetBstMidEndSepPunct{\mcitedefaultmidpunct}
{\mcitedefaultendpunct}{\mcitedefaultseppunct}\relax
\EndOfBibitem
\bibitem[Hordijk \latin{et~al.}(2017)Hordijk, Steel, and Dittrich]{hor17a}
Hordijk,~W.; Steel,~M.; Dittrich,~P. \emph{New J. Phys.} \textbf{2017},
  \emph{DOI:10.1088/1367-2630/aa9fcd}\relax
\mciteBstWouldAddEndPuncttrue
\mciteSetBstMidEndSepPunct{\mcitedefaultmidpunct}
{\mcitedefaultendpunct}{\mcitedefaultseppunct}\relax
\EndOfBibitem
\bibitem[Sousa \latin{et~al.}(2015)Sousa, Hordijk, Steel, and Martin]{sou}
Sousa,~F.~L.; Hordijk,~W.; Steel,~M.; Martin,~W. \emph{J. Syst. Chem.}
  \textbf{2015}, \emph{6}, 4\relax
\mciteBstWouldAddEndPuncttrue
\mciteSetBstMidEndSepPunct{\mcitedefaultmidpunct}
{\mcitedefaultendpunct}{\mcitedefaultseppunct}\relax
\EndOfBibitem
\bibitem[Vaidya \latin{et~al.}(2012)Vaidya, Manapat, A., Xulvi-Brunet, and
  Hayden]{vai}
Vaidya,~N.; Manapat,~M.~L.; A.,~C.~I.; Xulvi-Brunet,~R.; Hayden,~N.,~E.
  J.~Lehman \emph{Nature} \textbf{2012}, \emph{491}, 72--77\relax
\mciteBstWouldAddEndPuncttrue
\mciteSetBstMidEndSepPunct{\mcitedefaultmidpunct}
{\mcitedefaultendpunct}{\mcitedefaultseppunct}\relax
\EndOfBibitem
\bibitem[Gatti \latin{et~al.}(2017)Gatti, Hordijk, and Kauffman]{gat}
Gatti,~R.~C.; Hordijk,~W.; Kauffman,~S. \emph{Ecol. Model.} \textbf{2017},
  \emph{346}, 70--76\relax
\mciteBstWouldAddEndPuncttrue
\mciteSetBstMidEndSepPunct{\mcitedefaultmidpunct}
{\mcitedefaultendpunct}{\mcitedefaultseppunct}\relax
\EndOfBibitem
\bibitem[Gabora and Steel(2017)Gabora, and Steel]{gab}
Gabora,~L.; Steel,~M. \emph{J. Theor. Biol.} \textbf{2017}, \emph{63},
  617--638\relax
\mciteBstWouldAddEndPuncttrue
\mciteSetBstMidEndSepPunct{\mcitedefaultmidpunct}
{\mcitedefaultendpunct}{\mcitedefaultseppunct}\relax
\EndOfBibitem
\bibitem[Hordijk and Steel(2004)Hordijk, and Steel]{hor}
Hordijk,~W.; Steel,~M. \emph{J. Theor. Biol.} \textbf{2004}, \emph{227},
  451--461\relax
\mciteBstWouldAddEndPuncttrue
\mciteSetBstMidEndSepPunct{\mcitedefaultmidpunct}
{\mcitedefaultendpunct}{\mcitedefaultseppunct}\relax
\EndOfBibitem
\bibitem[Ashkenasy \latin{et~al.}(2004)Ashkenasy, Jegasia, Yadav, and
  Ghadiri]{ash}
Ashkenasy,~G.; Jegasia,~R.; Yadav,~M.; Ghadiri,~M.~R. \emph{Proc. Natl. Acad.
  Sci. USA} \textbf{2004}, \emph{101}, 10872--10877\relax
\mciteBstWouldAddEndPuncttrue
\mciteSetBstMidEndSepPunct{\mcitedefaultmidpunct}
{\mcitedefaultendpunct}{\mcitedefaultseppunct}\relax
\EndOfBibitem
\bibitem[Jain and Krishna(1998)Jain, and Krishna]{jai98}
Jain,~S.; Krishna,~S. \emph{Phys. Rev. Lett.} \textbf{1998}, \emph{81},
  5684--5687\relax
\mciteBstWouldAddEndPuncttrue
\mciteSetBstMidEndSepPunct{\mcitedefaultmidpunct}
{\mcitedefaultendpunct}{\mcitedefaultseppunct}\relax
\EndOfBibitem
\bibitem[Tarjan(1972)]{tar}
Tarjan,~R.~E. \emph{SIAM J. Comput.} \textbf{1972}, \emph{1}, 146--160\relax
\mciteBstWouldAddEndPuncttrue
\mciteSetBstMidEndSepPunct{\mcitedefaultmidpunct}
{\mcitedefaultendpunct}{\mcitedefaultseppunct}\relax
\EndOfBibitem
\bibitem[Hordijk \latin{et~al.}(2012)Hordijk, Steel, and Kauffman]{hor12}
Hordijk,~W.; Steel,~M.; Kauffman,~S. \emph{Acta Biotheor.} \textbf{2012},
  \emph{60}, 379--392\relax
\mciteBstWouldAddEndPuncttrue
\mciteSetBstMidEndSepPunct{\mcitedefaultmidpunct}
{\mcitedefaultendpunct}{\mcitedefaultseppunct}\relax
\EndOfBibitem
\bibitem[Hordijk and Steel(2013)Hordijk, and Steel]{hor13}
Hordijk,~W.; Steel,~M. \emph{J. Syst. Chem.} \textbf{2013}, \emph{4}, 3\relax
\mciteBstWouldAddEndPuncttrue
\mciteSetBstMidEndSepPunct{\mcitedefaultmidpunct}
{\mcitedefaultendpunct}{\mcitedefaultseppunct}\relax
\EndOfBibitem
\bibitem[Steel \latin{et~al.}(2013)Steel, Hordijk, and Smith]{ste13}
Steel,~M.; Hordijk,~W.; Smith,~J. \emph{J. Theor. Biol.} \textbf{2013},
  \emph{332}, 96--107\relax
\mciteBstWouldAddEndPuncttrue
\mciteSetBstMidEndSepPunct{\mcitedefaultmidpunct}
{\mcitedefaultendpunct}{\mcitedefaultseppunct}\relax
\EndOfBibitem
\bibitem[Bollobas and Rasmussen(1989)Bollobas, and Rasmussen]{bol}
Bollobas,~B.; Rasmussen,~S. \emph{Discrete Math.} \textbf{1989}, \emph{75},
  55--68\relax
\mciteBstWouldAddEndPuncttrue
\mciteSetBstMidEndSepPunct{\mcitedefaultmidpunct}
{\mcitedefaultendpunct}{\mcitedefaultseppunct}\relax
\EndOfBibitem
\bibitem[Cvetkovi{\'c} \latin{et~al.}(1995)Cvetkovi{\'c}, Doob, and Sachs]{cv}
Cvetkovi{\'c},~D.~M.; Doob,~M.; Sachs,~H. \emph{Spectra of Graphs, third ed.};
  Barth, Heidelberg, 1995\relax
\mciteBstWouldAddEndPuncttrue
\mciteSetBstMidEndSepPunct{\mcitedefaultmidpunct}
{\mcitedefaultendpunct}{\mcitedefaultseppunct}\relax
\EndOfBibitem
\end{mcitethebibliography}


\providecommand{\latin}[1]{#1}
\makeatletter
\providecommand{\doi}
  {\begingroup\let\do\@makeother\dospecials
  \catcode`\{=1 \catcode`\}=2\doi@aux}
\providecommand{\doi@aux}[1]{\endgroup\texttt{#1}}
\makeatother
\providecommand*\mcitethebibliography{\thebibliography}
\csname @ifundefined\endcsname{endmcitethebibliography}
  {\let\endmcitethebibliography\endthebibliography}{}

\end{document}